\documentclass[11pt,reqno]{amsart}
\usepackage{amsmath,amssymb,amsthm,bbm,mathtools,calc,verbatim,enumitem,tikz,url,mathrsfs,cite,fullpage}
\usepackage{float}
\usepackage[colorlinks,linkcolor=blue,citecolor=blue]{hyperref}
\usepackage{subcaption}

\usepackage{comment}
\usepackage{tabulary}
\usepackage{tabularx}
\usepackage{multirow}
\usepackage{booktabs}
\usepackage{cite}
\usepackage{enumitem}
\usepackage{float}
\usepackage[capitalize]{cleveref}

\newtheorem{theorem}{Theorem}[section]
\newtheorem{lemma}[theorem]{Lemma}

\theoremstyle{definition}

\theoremstyle{remark}

\newcommand{\prob}[3]{
	\begin{center}
		\fbox{~\begin{minipage}{.97\textwidth}
				\vspace{2pt}
				\noindent
				\normalsize\textsc{#1}
				
				\vspace{4pt}
				\setlength{\tabcolsep}{3pt}
				\renewcommand{\arraystretch}{1.0}
				\begin{tabularx}{\textwidth}{@{}lX@{}}
					\normalsize\textbf{Input:}	& \normalsize#2 \\
					\normalsize\textbf{Question:}		 & \normalsize#3
				\end{tabularx}
		\end{minipage}}
	\end{center}
}

\DeclareMathOperator{\vc}{vc}
\DeclareMathOperator{\fvs}{fvn}

\DeclareMathOperator{\tw}{tw}

\DeclareMathOperator{\mw}{mw}
\DeclareMathOperator{\pw}{pw}
\DeclareMathOperator{\nd}{nd}
\DeclareMathOperator{\td}{td}

\newcommand{\fm}{\textsc{Fair Matching}}
\newcommand{\fc}{\textsc{Fair Coloring}}

\newcommand{\umbp}{\textsc{Unary Multi-dimensional Bin Packing}}
\newcommand{\gt}{\textsc{Grid Tiling}}
\newcommand{\clique}{\textsc{Clique}}

\begin{document}

\title{Parameterized Complexity of Fair Coloring Problem}

\author{Ramin Javadi \and Hossein Shokouhi}

\address{Department of Mathematical Sciences, Isfahan University of Technology, P.O. Box 84156-83111, Isfahan, Iran.}\email{rjavadi@iut.ac.ir}

\begin{abstract}
Given a graph $G=(V,E)$, a (proper) $k$-coloring for $G$ is a vertex coloring with $k$ colors such that every two adjacent vertices receive different colors. Suppose that the vertex set $V$ is partitioned into some groups, a proper coloring is called fair if for every color class, the difference between the number of vertices in any two groups does not exceed a given threshold. In this paper, we investigate the parameterized complexity of the fair coloring problem with respect to the structural parameters of the input graph. 
In particular, we prove that the problem is W[1]-hard with respect to the number of groups for forests and also graphs of modular-width two, even when the number of colors is equal to two. On the positive side, we prove that when the number of colors is equal to two, then the problem is FPT with respect to  neighborhood diversity of the input graph. Moreover, in general, the problem is FPT with respect to neighborhood diversity and the number of groups. As a by-product, we prove that unary vector bin packing problem is W[1]-hard with respect to the dimension. 
\end{abstract}

\maketitle
\section{Introduction}
Given a graph $G=(V,E)$ and a positive integer $k$, a (proper) $k$-coloring for $G$ is a labeling of its vertices with $k$ colors such that every two adjacent vertices receive different colors. The set of all vertices of the same color is called a color class. In this paper, we study the problem of finding a fair coloring that satisfies a fairness Max-Min criterion. In particular, we suppose that  the vertex set $V$ is partitioned into $p$ groups and a positive integer $\ell$ is given.  We aim to find a proper $k$-coloring such that in each color class, the number of vertices of any two groups differs by at most $\ell$. We call such a coloring an $\ell$-fair coloring. 

Fairness has become a central concept in graph theory due to its ability to model the equitable allocation of resources, services, and opportunities over networked systems. Early work on fairness can be traced to classical graph optimization problems such as equitable graph coloring, introduced by Meyer \cite{meyer1973}, where the objective is to partition vertices into color classes of nearly equal size while maintaining a proper coloring. Since then, fairness has evolved into a broad research area encompassing fair graph partitioning, proportional matching, balanced separators, envy-free and equitable allocations on graphs, and fair network design. 
More recently, the emergence of algorithmic fairness and responsible artificial intelligence has renewed interest in fairness-aware graph algorithms, particularly for graph mining, social network analysis, graph neural networks, and recommendation systems, where fairness constraints are incorporated to reduce demographic bias while preserving algorithmic utility \cite{barocas2023}. From an algorithmic perspective, many fairness problems on graphs are NP-hard, motivating extensive research on exact, approximation, and parameterized algorithms and have found numerous applications in load balancing for distributed systems, communication and transportation networks, political districting, resource allocation, public service planning, clustering, and machine learning \cite{bera2019,barocas2023}.

\paragraph*{Our Contributions}
In this paper, we investigate the parameterized complexity of fair coloring problem with respect to different structural parameters of the input graph such as modular-width ($\mw$), tree-depth ($\td$), neighborhood diversity ($\nd$) and feedback vertex number ($\fvs$)  as well as the intrinsic parameters such as the number of colors and the number of groups. In particular, we prove that the problem is
FPT with respect to neighborhood diversity when $k=2$ (Theorem~\ref{thm:nd}) and with respect to neighborhood diversity and the number of groups, in general (Theorem~\ref{thm:p+nd}). We also prove that the problem is W[1]-hard with respect to the number of groups for forests and graphs with modular-width equal to two, even when $k=2$ (Theorem~\ref{thm:w1_p}).
The former result is obtained by a parameterized reduction from a problem called \umbp\ which we prove is W[1]-hard with respect to the dimension. Finally, we prove that the problem is in XP with respect to the number of groups when $k=2$ (Theorem~\ref{thm:xp}). A summary of our main results is depicted in Figure~\ref{fig:summary}.  For definitions of the  notions in parameterized complexity, see \cite{cygan}. Also, for the definition of structural parameters of graphs see \cite{lampis}.

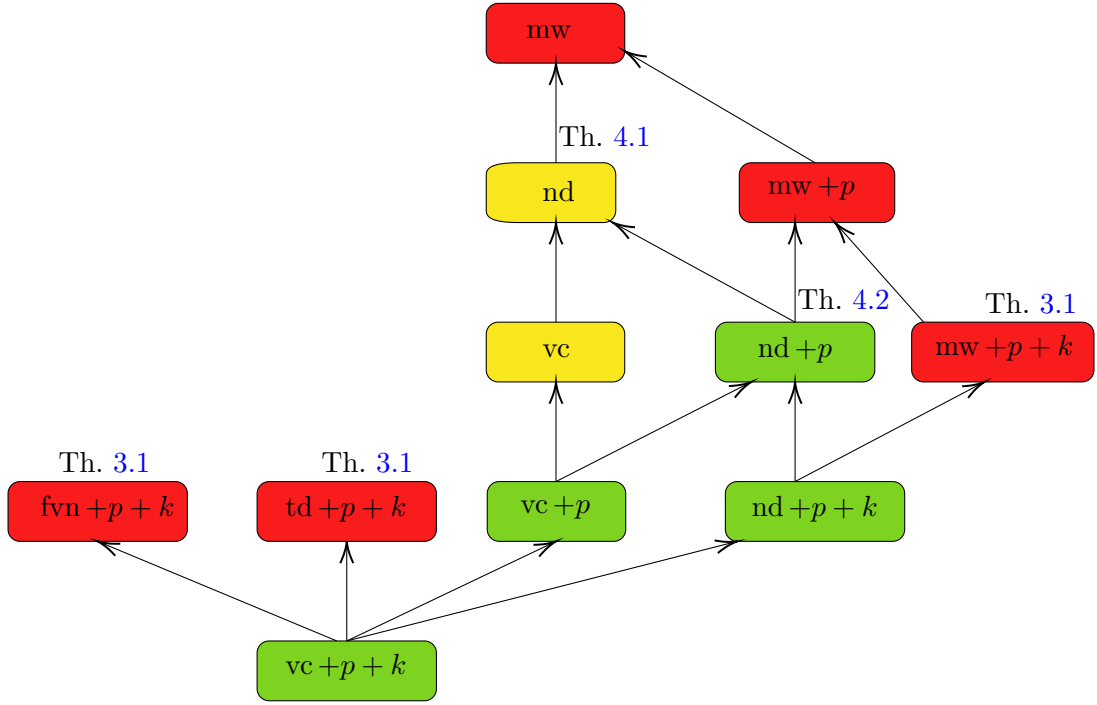
\begin{figure}[ht]
	\begin{center}   
		
		\begin{tikzpicture}[x=0.75pt,y=0.75pt,yscale=-1,xscale=1]
		
		\draw    (280,420) -- (280,372) ;
		\draw [shift={(280,370)}, rotate = 90] [color={rgb, 255:red, 0; green, 0; blue, 0 }  ][line width=0.75]    (10.93,-3.29) .. controls (6.95,-1.4) and (3.31,-0.3) .. (0,0) .. controls (3.31,0.3) and (6.95,1.4) .. (10.93,3.29)   ;
		\draw  [fill={rgb, 255:red, 126; green, 211; blue, 33 }  ,fill opacity=1 ] (235,426) .. controls (235,422.69) and (237.69,420) .. (241,420) -- (319,420) .. controls (322.31,420) and (325,422.69) .. (325,426) -- (325,444) .. controls (325,447.31) and (322.31,450) .. (319,450) -- (241,450) .. controls (237.69,450) and (235,447.31) .. (235,444) -- cycle ;
		\draw  [fill={rgb, 255:red, 126; green, 211; blue, 33 }  ,fill opacity=1 ] (350.5,346) .. controls (350.5,342.69) and (353.19,340) .. (356.5,340) -- (414,340) .. controls (417.31,340) and (420,342.69) .. (420,346) -- (420,364) .. controls (420,367.31) and (417.31,370) .. (414,370) -- (356.5,370) .. controls (353.19,370) and (350.5,367.31) .. (350.5,364) -- cycle ;
		\draw  [fill={rgb, 255:red, 248; green, 231; blue, 28 }  ,fill opacity=1 ] (350,266) .. controls (350,262.69) and (352.69,260) .. (356,260) -- (413.5,260) .. controls (416.81,260) and (419.5,262.69) .. (419.5,266) -- (419.5,284) .. controls (419.5,287.31) and (416.81,290) .. (413.5,290) -- (356,290) .. controls (352.69,290) and (350,287.31) .. (350,284) -- cycle ;
		\draw  [fill={rgb, 255:red, 248; green, 231; blue, 28 }  ,fill opacity=1 ] (350,186) .. controls (350,182.69) and (352.69,180) .. (366,180) -- (409,180) .. controls (412.31,180) and (415,182.69) .. (415,186) -- (415,204) .. controls (415,207.31) and (412.31,210) .. (409,210) -- (366,210) .. controls (352.69,210) and (350,207.31) .. (350,204) -- cycle ;
		\draw  [fill={rgb, 255:red, 248; green, 28; blue, 28 }  ,fill opacity=1 ] (350,106) .. controls (350,102.69) and (352.69,100) .. (356,100) -- (413.5,100) .. controls (416.81,100) and (419.5,102.69) .. (419.5,106) -- (419.5,124) .. controls (419.5,127.31) and (416.81,130) .. (413.5,130) -- (356,130) .. controls (352.69,130) and (350,127.31) .. (350,124) -- cycle ;
		\draw  [fill={rgb, 255:red, 126; green, 211; blue, 33 }  ,fill opacity=1 ] (470,346) .. controls (470,342.69) and (472.69,340) .. (476,340) -- (554,340) .. controls (557.31,340) and (560,342.69) .. (560,346) -- (560,364) .. controls (560,367.31) and (557.31,370) .. (554,370) -- (476,370) .. controls (472.69,370) and (470,367.31) .. (470,364) -- cycle ;
		\draw  [fill={rgb, 255:red, 248; green, 28; blue, 28 }  ,fill opacity=1 ] (110,346) .. controls (110,342.69) and (112.69,340) .. (116,340) -- (194,340) .. controls (197.31,340) and (200,342.69) .. (200,346) -- (200,364) .. controls (200,367.31) and (197.31,370) .. (194,370) -- (116,370) .. controls (112.69,370) and (110,367.31) .. (110,364) -- cycle ;
		\draw  [fill={rgb, 255:red, 248; green, 28; blue, 28 }  ,fill opacity=1 ] (477,186) .. controls (477,182.69) and (479.69,180) .. (483,180) -- (549,180) .. controls (552.31,180) and (555,182.69) .. (555,186) -- (555,204) .. controls (555,207.31) and (552.31,210) .. (549,210) -- (483,210) .. controls (479.69,210) and (477,207.31) .. (477,204) -- cycle ;
		\draw  [fill={rgb, 255:red, 248; green, 28; blue, 28 }  ,fill opacity=1 ] (235,346) .. controls (235,342.69) and (237.69,340) .. (241,340) -- (319,340) .. controls (322.31,340) and (325,342.69) .. (325,346) -- (325,364) .. controls (325,367.31) and (322.31,370) .. (319,370) -- (241,370) .. controls (237.69,370) and (235,367.31) .. (235,364) -- cycle ;
		\draw  [fill={rgb, 255:red, 126; green, 211; blue, 33 }  ,fill opacity=1 ] (465,266) .. controls (465,262.69) and (467.69,260) .. (471,260) -- (539,260) .. controls (542.31,260) and (545,262.69) .. (545,266) -- (545,284) .. controls (545,287.31) and (542.31,290) .. (539,290) -- (471,290) .. controls (467.69,290) and (465,287.31) .. (465,284) -- cycle ;
		\draw  [fill={rgb, 255:red, 248; green, 28; blue, 28 }  ,fill opacity=1 ] (563.5,266) .. controls (563.5,262.69) and (566.19,260) .. (569.5,260) -- (649,260) .. controls (652.31,260) and (655,262.69) .. (655,266) -- (655,284) .. controls (655,287.31) and (652.31,290) .. (649,290) -- (569.5,290) .. controls (566.19,290) and (563.5,287.31) .. (563.5,284) -- cycle ;
		\draw    (505,260) -- (505,212) ;
		\draw [shift={(505,210)}, rotate = 90] [color={rgb, 255:red, 0; green, 0; blue, 0 }  ][line width=0.75]    (10.93,-3.29) .. controls (6.95,-1.4) and (3.31,-0.3) .. (0,0) .. controls (3.31,0.3) and (6.95,1.4) .. (10.93,3.29)   ;
		\draw    (505,340) -- (505,292) ;
		\draw [shift={(505,290)}, rotate = 90] [color={rgb, 255:red, 0; green, 0; blue, 0 }  ][line width=0.75]    (10.93,-3.29) .. controls (6.95,-1.4) and (3.31,-0.3) .. (0,0) .. controls (3.31,0.3) and (6.95,1.4) .. (10.93,3.29)   ;
		\draw    (385,180) -- (385,132) ;
		\draw [shift={(385,130)}, rotate = 90] [color={rgb, 255:red, 0; green, 0; blue, 0 }  ][line width=0.75]    (10.93,-3.29) .. controls (6.95,-1.4) and (3.31,-0.3) .. (0,0) .. controls (3.31,0.3) and (6.95,1.4) .. (10.93,3.29)   ;
		\draw    (385,260) -- (385,212) ;
		\draw [shift={(385,210)}, rotate = 90] [color={rgb, 255:red, 0; green, 0; blue, 0 }  ][line width=0.75]    (10.93,-3.29) .. controls (6.95,-1.4) and (3.31,-0.3) .. (0,0) .. controls (3.31,0.3) and (6.95,1.4) .. (10.93,3.29)   ;
		\draw    (385,340) -- (385,292) ;
		\draw [shift={(385,290)}, rotate = 90] [color={rgb, 255:red, 0; green, 0; blue, 0 }  ][line width=0.75]    (10.93,-3.29) .. controls (6.95,-1.4) and (3.31,-0.3) .. (0,0) .. controls (3.31,0.3) and (6.95,1.4) .. (10.93,3.29)   ;
		\draw    (280,420) -- (383.19,370.86) ;
		\draw [shift={(385,370)}, rotate = 154.54] [color={rgb, 255:red, 0; green, 0; blue, 0 }  ][line width=0.75]    (10.93,-3.29) .. controls (6.95,-1.4) and (3.31,-0.3) .. (0,0) .. controls (3.31,0.3) and (6.95,1.4) .. (10.93,3.29)   ;
		\draw    (280,420) -- (474.06,370.49) ;
		\draw [shift={(476,370)}, rotate = 165.69] [color={rgb, 255:red, 0; green, 0; blue, 0 }  ][line width=0.75]    (10.93,-3.29) .. controls (6.95,-1.4) and (3.31,-0.3) .. (0,0) .. controls (3.31,0.3) and (6.95,1.4) .. (10.93,3.29)   ;
		\draw    (505,340) -- (598.23,290.93) ;
		\draw [shift={(600,290)}, rotate = 152.24] [color={rgb, 255:red, 0; green, 0; blue, 0 }  ][line width=0.75]    (10.93,-3.29) .. controls (6.95,-1.4) and (3.31,-0.3) .. (0,0) .. controls (3.31,0.3) and (6.95,1.4) .. (10.93,3.29)   ;
		\draw    (275,420) -- (156.85,370.77) ;
		\draw [shift={(155,370)}, rotate = 22.62] [color={rgb, 255:red, 0; green, 0; blue, 0 }  ][line width=0.75]    (10.93,-3.29) .. controls (6.95,-1.4) and (3.31,-0.3) .. (0,0) .. controls (3.31,0.3) and (6.95,1.4) .. (10.93,3.29)   ;
		\draw    (569.5,260) -- (526.33,211.49) ;
		\draw [shift={(525,210)}, rotate = 48.33] [color={rgb, 255:red, 0; green, 0; blue, 0 }  ][line width=0.75]    (10.93,-3.29) .. controls (6.95,-1.4) and (3.31,-0.3) .. (0,0) .. controls (3.31,0.3) and (6.95,1.4) .. (10.93,3.29)   ;
		\draw    (385,340) -- (481.22,290.91) ;
		\draw [shift={(483,290)}, rotate = 152.97] [color={rgb, 255:red, 0; green, 0; blue, 0 }  ][line width=0.75]    (10.93,-3.29) .. controls (6.95,-1.4) and (3.31,-0.3) .. (0,0) .. controls (3.31,0.3) and (6.95,1.4) .. (10.93,3.29)   ;
		\draw    (505,260) -- (415.26,210.96) ;
		\draw [shift={(413.5,210)}, rotate = 28.65] [color={rgb, 255:red, 0; green, 0; blue, 0 }  ][line width=0.75]    (10.93,-3.29) .. controls (6.95,-1.4) and (3.31,-0.3) .. (0,0) .. controls (3.31,0.3) and (6.95,1.4) .. (10.93,3.29)   ;
		\draw    (515,180) -- (421.23,125.01) ;
		\draw [shift={(419.5,124)}, rotate = 30.39] [color={rgb, 255:red, 0; green, 0; blue, 0 }  ][line width=0.75]    (10.93,-3.29) .. controls (6.95,-1.4) and (3.31,-0.3) .. (0,0) .. controls (3.31,0.3) and (6.95,1.4) .. (10.93,3.29)   ;
		
		\draw (248,346) node [anchor=north west][inner sep=0.75pt]    {$\td+p+k$};
		\draw (485,266) node [anchor=north west][inner sep=0.75pt]    {$\nd+p$};
		\draw (490,186) node [anchor=north west][inner sep=0.75pt]    {$\mw+p$};
		\draw (248,425) node [anchor=north west][inner sep=0.75pt]    {$\vc+p+k$};
		\draw (482,346) node [anchor=north west][inner sep=0.75pt]    {$\nd+p+k$};
		\draw (574,265) node [anchor=north west][inner sep=0.75pt]    {$\mw+p+k$};
		\draw (366,346) node [anchor=north west][inner sep=0.75pt]    {$\vc+p$};
		\draw (377,270) node [anchor=north west][inner sep=0.75pt]    {$\vc$};
		\draw (377,188) node [anchor=north west][inner sep=0.75pt]    {$\nd$};
		\draw (369,110) node [anchor=north west][inner sep=0.75pt]    {$\mw$};
		\draw (125,346) node [anchor=north west][inner sep=0.75pt]    {$\fvs+p+k$};
		\draw (505,242) node [anchor=north west][inner sep=0.75pt]    {Th.~\ref{thm:p+nd}};
		\draw (266,324) node [anchor=north west][inner sep=0.75pt]    {Th.~\ref{thm:w1_p}};
		\draw (134,324) node [anchor=north west][inner sep=0.75pt]    {Th.~\ref{thm:w1_p}};
		\draw (599,244) node [anchor=north west][inner sep=0.75pt]    {Th.~\ref{thm:w1_p}};
		\draw (385,160) node [anchor=north west][inner sep=0.75pt]    {Th.~\ref{thm:nd}};

		\end{tikzpicture}
	\end{center}
	\caption{Summary of our main results regarding \fc\ problem. An arrow from $f$ to $g$ means that $g$ is bounded by a function of $f$ and so W[1]-hardness result with respect to $f$ implies W[1]-hardness with respect to $g$. Parameters marked by green are proved to be FPT. Parameters marked by red are proved to be W[1]-hard. Parameters marked by yellow is known to be FPT in the case of $k=2$ and its complexity is unknown in general case.} \label{fig:summary}
\end{figure}
\paragraph*{Related Work}
Parameterized complexity of coloring problems has been extensively investigated in the literature. Finding a proper $k$-coloring is known to be NP-hard even when $k=3$ on general graphs \cite{garey76}. The problem is FPT with respect to treewidth \cite{bodlaender1988} and neighborhood diversity of the input graph \cite{ganian2012}. The problem is known to be W[1]-hard and also in XP with respect to cliquewidth \cite{fomin2009,rao2007}.
Some extensions of the coloring problem have been investigated such as equitable coloring, precoloring extension and list coloring problems. These three problems are known to be W[1]-hard with respect to treewidth \cite{fellows2011}. The former problem is also W[1]-hard with respect to vertex cover and the other two problems are FPT with respect to vertex cover \cite{fiala2011}. In \cite{belmonte2020}, a variant of coloring problem called defective coloring is investigated in which, given integers $k$, $\Delta$, we are looking for a $k$-coloring such that the induced graph on each color class has maximum degree $\Delta$. They showed that the problem is W[1]-hard with respect to tree-depth and feedback vertex set when $k=2$. 

The study of fairness notion in coloring problems dates back to 1973 in \cite{meyer1973} where the \textsc{Equitable Coloring} problem is investigated. In this problem, we seek for a proper coloring of vertices such that the size of color classes are almost equal, i.e. 
the difference between the size of color classes is at most one. This fairness measure is called \textit{Max-Min criterion}. There are other fairness measures like \textit{margin of victory} (MoV) which is defined as the difference of the most and the second frequent class. The MoV measure is studied e.g. in \textsc{Fair Connected Districting} problem in which we are given an integer $\ell$, a graph and a partition of vertices into $p$ groups and we are looking for a partition of vertices into $k$ districts (connected subgraphs) such that each district is $\ell$-fair with MoV measure, i.e. in each district, the difference of the most and the second frequent groups is at most $\ell$. This problem is firstly introduced by Stoica et al. \cite{stoica} and, the parameterized complexity of the problem is investigated in \cite{boehmer23}. In particular, it is proved that the problem is W[1]-hard with respect to treewidth, $k$ and $p$, combined. Also, it is W[1]-hard with respect to feedback edge number and $k$, combined. It is also FPT with respect to vertex cover number and $p$, combined. 

The Max-Min fairness measure is considered e.g. in \textsc{Fair Shortest Path} problem \cite{bentert} in which we are given a graph $G$, a  partition of vertices into $p$ groups, two vertices $s,t$ and two integers $\ell, \delta$ and the question is that if there exists an $s-t-$path $P$ of length at most $\ell$ such that the difference of the numbers of vertices in $P$ in each two groups is bounded by $\delta$. Bentert et al. \cite{bentert} proved that \textsc{Fair Shortest Path} is in XP and W[1]-hard with respect to $p$ and is FPT with respect to $\ell$ and $\delta$. 

Fairness is also studied in matching problems. Given a bipartite graph $G=(U\cup V,E)$, a left-perfect many-to-one matching is a subset $M \subseteq E$ such that each vertex in $U$ is incident with exactly one edge in $M$. In \fm\ problem,  given an integer $\ell$, a bipartite graph $G=(U\cup V,E)$ and a partition of $U$ into $p$ groups, we are looking for a left-perfect many-to-one matching which is $\ell$-fair, i.e. the neighborhood of each vertex in $V$ is $\ell$-fair with respect to Max-Min or MoV criteria. In \cite{boehmer24}, it is proved that \fm\ is NP-hard for $p=3$ and $\Delta=5$ and FPT with respect to $|V|$. In \cite{javadi25}, it is proved that the problem is W[1]-hard with respect to $\fvs$, $\td+\Delta_U$, and $\pw+p+\Delta_U$. On the positive side, it is proved that the problem is FPT with respect to $\tw+\Delta_V$, $\nd$, and $\td+ p$. 

\section{Problem Formulation and Preliminaries}
For positive integers $m,n,m < n $, let us define $[m,n]=\{m,m+1,\dots,n\}$ and $[n]=[1,n]$. For a graph $G=(V,E)$ and two subsets $A,B \subseteq E$ the set $E(A,B)$ denotes the set of all edges in $E$ with one endpoint in $A$ and one endpoint in $B$. Also, for a vertex $v \in V$, we say that $v$ is complete (resp. incomplete) to $A$ if $v$ is adjacent (resp. nonadjacent) to all vertices in $A$.

Let $G=(V,E) $ be a graph and $k$ be a positive integer. A (proper) $k$-coloring for $G$ is a function $c:V(G)\to [k]$ such that for any two adjacent vertices $u$ and $v$, we have $c(u)\neq c(v)$. For each $i\in [k]$, the color class $C_i=c^{-1}(i)$ is defined as the vertices whose color is $i$ which is a stable set. 

Now, let $(V_1,\ldots, V_p)$ be a partition of $V(G)$ into $p$ non-empty subsets. Now, let $\ell$ be an integer. We say that a subset $U\subseteq V(G)$ is \textit{$\ell$-fair} if  $\max_{j\in [p]} |U\cap V_j|-\min_{j\in [p]} |U\cap V_j|\leq \ell  $.
A coloring $c:V(G)\to [k]$ is called \textit{$\ell$-fair} if for each $i\in [k]$, the color class $C_i$ is $\ell$-fair. 

\prob{Fair Coloring}{A graph $G=(V,E)$, a partition $(V_1,\ldots,V_p)$ of $V$ into non-empty groups, three positive integers $k$ and a non-negative integer $\ell$.}{Is there a proper $k$-coloring for $G$ which is $\ell$-fair?}

We need the idea of balanced coloring in the proof of Theorem~\ref{thm:p+nd}. Let $G$ be a multigraph and $c:E(G)\to \{1,\ldots, k\}$ be an edge coloring of $G$. For each $i\in [k]$ and every pair of vertices $u,v\in V(G)$, define $C_i(v)$ and $C_i(u,v)$ to be the set of all edges of color $i$ incident with $v$, and the set of all edges of color $i$ with endpoints $u,v$, respectively. We say that $c$ is a \textit{balanced coloring}, if for every vertex $v\in V(G)$, 
\[
\max_{1\leq i<j\leq k} \big| |C_i(v)|-|C_j(v)|\big|\leq 1,
\] 
and for every pair of vertices $u,v \in V(G)$,
\[
\max_{1\leq i<j\leq k} \big| |C_i(u,v)|-|C_j(u,v)|\big|\leq 1.
\] 
The following lemma guarantees existence of a balanced coloring for every bipartite graph $G$ and every positive integer $k$.

\begin{lemma}{\rm \cite{dewerra75,dewerra71}} \label{lem:dewerra}
	For every bipartite multigraph $G$ and every positive integer $k$, $G$ admits a balanced $k$-edge coloring. 
\end{lemma}

Lemma~\ref{lem:dewerra} can be translated into matrix language in this setting. 
\begin{lemma}\label{lem:balanced}
	Let $A$ be an $m\times n$ matrix with non-negative integer entries and $b$ be a positive integer. Then, $A$ can be written as the sum of $b$ matrices $A^1+\cdots+A^b$ with non-negative integer entries, such that for each row $r\in [m]$, each column $c\in [n]$, and each $i,j\in [b]$, we have 
	\[
	|(A^i)_{r}- (A^j)_{r}|\leq 1,
	\] 
	and
	\[
	|(A^i)^{c}- (A^j)^{c}|\leq 1,
	\] 
	where $(A^i)_{r}$ is the sum of entries of row $r$ in $A^i$ and $(A^i)^{c}$ is the sum of entries of column $c$ in $A^i$.
\end{lemma}

We also need the following classical theorem by Lenstra about solving integer programming.

\begin{theorem} {\rm \cite{lenstra}} \label{thm:lenstra}
	Every integer programming with $k$ variables can be solved in time $O^*(k^{O(k)})$.
\end{theorem}

\section{W[1]-hard Results}

In this section, we prove that \fc\ is  W[1]-hard with respect to the number of groups $p$ even when the number of colors $k$ is equal to two on forests and graphs with modular-width equal to two. In order to prove this result, we used an intermediate problem called unary vector bin packing in which each item has $d$ different weights (dimensions) and 
selected items satisfy $d$ capacity constraints in the $d$ dimensions. We will prove that the problem is W[1]-hard with respect to $d$ even when the weights are given in unary encoding. 

\begin{theorem} \label{thm:w1_p}
	\fc\ is W[1]-hard with respect to $p$ even when $k=2$ for the following input graphs:
	\begin{itemize}
		\item[\rm (i)] graphs of modular-width $2$.
		\item[\rm (ii)] forests whose every component has depth equal to $2$.
	\end{itemize} 
\end{theorem}

In order to prove the above theorem, we need an intermediate problem, called \umbp\ or {\sc Unary Vector Bin Packing} (see e.g. \cite{woeginger97,christensen16,chekuri04}). 

\prob{Unary Multi-dimensional (Vector) Bin Packing}{$(\mathcal{I},\{\mathbf{w}_i\}_{i\in \mathcal{I}}, \mathbf{B},m,d)$. Positive integers $m,d$, a set of items $\mathcal{I}$, where each item $i\in \mathcal{I}$ has a vector weight $\mathbf{w}_i=(w_{i,1},\ldots, w_{i,d}) \in (\mathbb{Z}^+)^{d}$, and a capacity vector $\mathbf{B}=(B_1,\ldots, B_d)\in (\mathbb{Z}^+)^{d}$. All weights and capacities are given in unary encoding.}{Is there a subset $I\subseteq \mathcal{I}$, such that $|I|=m$ and for each $j\in [d]$, we have $\sum_{i\in I}w_{i,j}=B_j$ (called capacity constraints)?}

Note that the case $d=1$ is the well-known Knapsack problem which is known to be NP-hard when the input is given in binary encoding and can be solved in polynomial time using dynamic programming, when the input is unary encoding (see \cite{garey76}). 
Here, we first prove that \umbp\ is W[1]-hard with respect to $d$ and $m$. 

\begin{theorem} \label{thm:umbp}
	\umbp\ is W[1]-hard with respect to $d$ and $m$, combined. 
\end{theorem}

\begin{proof}
	We give a parameterized reduction from \gt. 
	
	\prob{Grid Tiling}{
		A positive integer $t$. For each $i,j\in [t]$, a subset $S_{i,j}\subseteq [N]\times [N] $. 
	}{Can we choose one pair $(x_{i,j},y_{i,j})\in S_{i,j}$ for each $i,j\in [t]$ such that $x_{1,j}=x_{2,j}=\cdots=x_{t,j}$ and $y_{i,1}=y_{i,2}=\cdots=y_{i,t}$?}
	
	It is known that \gt\ is W[1]-hard with respect to $t$ \cite{marx07,cygan}. Also, note that \gt\ remains W[1]-hard even when the input (the elements of $S_{i,j}$'s) is given in unary encoding, since the reduction is from \clique\  (see \cite{cygan}). 
	
	Let $\mathcal{S}=(S_{i,j})_{i,j\in [t]}$ be an instance of \gt. We construct an instance of \umbp\ as follows. For each $(i,j)\in [t]\times [t]$ and each $(x,y)\in S_{i,j}$ we have a distinct item $(i,j;x,y)$. So, 
	$$\mathcal{I}=\{(i,j;x,y)\mid (i,j)\in  [t]\times [t], (x,y)\in S_{i,j} \}.$$
	
	Also, the dimension is equal to $d=t^2+2t(t-1)$. For each pair $(i',j')\in [t]\times [t]$, we have one dimension $(i',j')^E$ (corresponding to entries), for each pair $(i',j')\in [2,t]\times [t]$, we have one dimension $(i',j')^C$ (corresponding to columns) and for each pair $(i',j')\in [t]\times [2,t]$, we have one dimension $(i',j')^R$ (corresponding to rows). For each item $(i,j;x,y)$ and each $(i',j')$, the weights are defined as follows.
	
	\begin{align*}
	w_{(i,j;x,y),(i',j')^E} &= \begin{cases}
	1 & (i',j')= (i,j),\\
	0 & \text{otherwise},
	\end{cases} \\
	w_{(i,j;x,y),(i',j')^C}& = \begin{cases}
	x & (i',j')= (i,j),\\
	N-x & (i',j')= (i+1,j),\\
	0 & \text{otherwise},
	\end{cases}\\
	w_{(i,j;x,y),(i',j')^R} &= \begin{cases}
	y & (i',j')= (i,j),\\
	N-y & (i',j')= (i,j+1),\\
	0 & \text{otherwise}.
	\end{cases}\\
	\end{align*}
	Also, for each $(i,j)$, we define $B_{(i,j)^E}=1$ and $B_{(i,j)^C}=B_{(i,j)^R}=N$. Also, we take $m=t^2$. \\
	
	\noindent\textbf{Proof of Sufficiency.}
	Suppose that $\mathcal{S}$ is a yes-instance for \gt. So, for each $(i,j)\in [t]\times [t]$, we can choose $(x_{i,j},y_{i,j})\in S_{i,j}$ such that  $x_{1,j}=x_{2,j}=\cdots=x_{t,j}$ and $y_{i,1}=y_{i,2}=\cdots=y_{i,t}$. Now, define
	
	\[
	\mathcal{I}:= \{ (i,j;x_{i,j},y_{i,j})\mid (i,j)\in [t]\times [t]\}.
	\]
	It is clear that $|I|=t^2$. Also, since for each $(i,j)\in [t]\times [t]$, we just have one item $(i,j;x,y)$ in $I$, so $\sum_{a\in I} w_{a,(i,j)^E}=1 =B_{(i,j)^E}$. On the other hand, for each $(i,j)\in [2,t]\times [t]$, we have 
	$\sum_{a\in I} w_{a,(i,j)^C}= x_{i,j}+ N- x_{i-1,j} = N=  B_{(i,j)^C}$. Similarly, for each $(i,j)\in [t]\times [2,t]$, we have 
	$\sum_{a\in I} w_{a,(i,j)^R}= y_{i,j}+ N- y_{i,j-1} = N=  B_{(i,j)^R}$. Hence, the capacity constraints in all dimensions are satisfied and $(\mathcal{I},\{\mathbf{w}_a\}_{a\in \mathcal{I}}, \mathbf{B},m,d) $ is a yes-instance for \umbp. \\

	\noindent\textbf{Proof of Necessity.}
	Now, suppose that $(\mathcal{I},\{\mathbf{w}_a\}_{a\in \mathcal{I}}, \mathbf{B},t^2) $ is a yes-instance for \umbp, so there is an $I\subseteq \mathcal{I}$ with $|I|=t^2$ which satisfies capacity constraints. Fix an  $(i,j)\in [t]\times [t]$. We have $\sum_{a\in I}
	w_{a,(i,j)^E}= B_{(i,j)^E}=1$. So, there is a unique item $(i,j,x_{i,j},y_{i,j})$ in $I$. 
	We take $(x_{i,j}, y_{i,j})$ as the solution of \gt. Now, for each $(i,j)\in [2,t]\times [t]$, we have
	$$ \sum_{a\in I}
	w_{a,(i,j)^C}=x_{i,j}+N-x_{i-1,j} = B_{(i,j)^C}=N.$$
	Therefore, $x_{i,j}=x_{i-1,j}$ as desired. Similarly, for each  $(i,j)\in [t]\times [2,t]$, we have $y_{i,j}=y_{i,j-1}$. Hence, $\mathcal{S}$ is a yes-instance for \gt.   
	
	Finally, the dimension $d=3t^2-2t$ and $|I|=t^2$ are functions of $t$. So, this is a parameterized reduction. This completes the proof.
\end{proof}

Now, we are ready to prove Theorem~\ref{thm:w1_p}.
\begin{proof}[Proof of Theorem~\ref{thm:w1_p}]
	We give a parameterized reduction from \umbp\ which is W[1]-hard with respect to $d$ (Theorem~\ref{thm:umbp}). 
	Let $(\mathcal{I},\{\mathbf{w}_i\}_{i\in \mathcal{I}}, \mathbf{B},m,d)$ be an instance of \umbp, where $|\mathcal{I}|=n$ and $B=\max_{j\in [d]}B_j$. Also, for each $j\in [d]$, define $A_j= \sum_{i\in \mathcal{I}} w_{i,j}$ and $A=\max_{j\in [d]} A_j$. 
	
	We construct an instance of \fc\ as follows. Let $p=d+1$ and $\ell=0$. Let $S$ and $T$ be sets with $|S|=n-m+B$ and $|T|=A+m+1$. Also, for each $(i,j)\in [n]\times [d]$, let $W_{i,j}$ be a set with $|W_{i,j}|=w_{i,j}$, $X_j,Y_j$ be sets with $|X_j|=B-B_j$, $|Y_j|=B_j+A-A_j$. Also, let $Z=\{z_{i,j}\mid i\in [n], j\in [d]\}$. All these sets are disjoint. The graph $G=(V,E)$ is defined as a bipartite graph such that $V$ is partitioned into $(V_1,\ldots ,V_p)$, where for each $j\in [d]$,
	$$V_j=(\cup_{i\in \mathcal{I}} W_{i,j}) \cup X_j\cup Y_j\cup \{z_{i,j}: i\in [n]\}.$$
	Also, 
	$$V_{d+1}=S\cup T.$$
	The graph $G$ has a bipartition $(X,Y)$, where $X=S\cup_{j\in [d]} X_j\cup_{i\in [n],j\in [d]} W_{i,j} $ and $Y= T\cup_{j\in [d]} Y_j\cup Z $. Also, the graph $G$ has $n+1$ connected components, one component on the set $S\cup T\cup (\cup_{j\in [d]} X_j\cup Y_j)$ and one component on the set $\cup_{j\in [d]} W_{i,j} \cup \{z_{i,j}: j\in [d]\}$, for each $i\in [n]$. Each connected component can induce either a complete bipartite graph or a tree of depth two (depending on whether we want to prove Item (i) or (ii)).  
	Also, let $k=2$ and $\ell=0$.\\
	
	\begin{figure}
		\begin{center}

			\tikzset{every picture/.style={line width=0.75pt}} 
			
			\begin{tikzpicture}[x=0.75pt,y=0.75pt,yscale=-1,xscale=1]
			
			\draw  [fill={rgb, 255:red, 255; green, 255; blue, 255 }  ,fill opacity=1 ] (24,480) .. controls (24,475.58) and (27.58,472) .. (32,472) -- (56,472) .. controls (60.42,472) and (64,475.58) .. (64,480) -- (64,504) .. controls (64,508.42) and (60.42,512) .. (56,512) -- (32,512) .. controls (27.58,512) and (24,508.42) .. (24,504) -- cycle ;
			\draw  [fill={rgb, 255:red, 255; green, 255; blue, 255 }  ,fill opacity=1 ] (188,576) .. controls (188,571.58) and (191.58,568) .. (196,568) -- (220,568) .. controls (224.42,568) and (228,571.58) .. (228,576) -- (228,600) .. controls (228,604.42) and (224.42,608) .. (220,608) -- (196,608) .. controls (191.58,608) and (188,604.42) .. (188,600) -- cycle ;
			\draw  [fill={rgb, 255:red, 255; green, 255; blue, 255 }  ,fill opacity=1 ] (24,576) .. controls (24,571.58) and (27.58,568) .. (32,568) -- (56,568) .. controls (60.42,568) and (64,571.58) .. (64,576) -- (64,600) .. controls (64,604.42) and (60.42,608) .. (56,608) -- (32,608) .. controls (27.58,608) and (24,604.42) .. (24,600) -- cycle ;
			\draw  [fill={rgb, 255:red, 255; green, 255; blue, 255 }  ,fill opacity=1 ] (88,575) .. controls (88,570.58) and (91.58,567) .. (96,567) -- (120,567) .. controls (124.42,567) and (128,570.58) .. (128,575) -- (128,599) .. controls (128,603.42) and (124.42,607) .. (120,607) -- (96,607) .. controls (91.58,607) and (88,603.42) .. (88,599) -- cycle ;
			\draw  [fill={rgb, 255:red, 255; green, 255; blue, 255 }  ,fill opacity=1 ] (88,480) .. controls (88,475.58) and (91.58,472) .. (96,472) -- (120,472) .. controls (124.42,472) and (128,475.58) .. (128,480) -- (128,504) .. controls (128,508.42) and (124.42,512) .. (120,512) -- (96,512) .. controls (91.58,512) and (88,508.42) .. (88,504) -- cycle ;
			\draw  [fill={rgb, 255:red, 255; green, 255; blue, 255 }  ,fill opacity=1 ] (187,481) .. controls (187,476.58) and (190.58,473) .. (195,473) -- (219,473) .. controls (223.42,473) and (227,476.58) .. (227,481) -- (227,505) .. controls (227,509.42) and (223.42,513) .. (219,513) -- (195,513) .. controls (190.58,513) and (187,509.42) .. (187,505) -- cycle ;
			\draw    (43,512) -- (43,568) ;
			\draw    (43,512) -- (107.19,567.45) ;
			\draw    (43,512) -- (208,568) ;
			\draw    (196,513) -- (43,568) ;
			\draw    (112,512) -- (43,568) ;
			\draw  [fill={rgb, 255:red, 0; green, 0; blue, 0 }  ,fill opacity=1 ] (261,587) .. controls (261,584.24) and (263.24,582) .. (266,582) .. controls (268.76,582) and (271,584.24) .. (271,587) .. controls (271,589.76) and (268.76,592) .. (266,592) .. controls (263.24,592) and (261,589.76) .. (261,587) -- cycle ;
			\draw  [fill={rgb, 255:red, 0; green, 0; blue, 0 }  ,fill opacity=1 ] (356,590) .. controls (356,587.24) and (358.24,585) .. (361,585) .. controls (363.76,585) and (366,587.24) .. (366,590) .. controls (366,592.76) and (363.76,595) .. (361,595) .. controls (358.24,595) and (356,592.76) .. (356,590) -- cycle ;
			\draw  [fill={rgb, 255:red, 255; green, 255; blue, 255 }  ,fill opacity=1 ] (244,481) .. controls (244,476.58) and (247.58,473) .. (252,473) -- (276,473) .. controls (280.42,473) and (284,476.58) .. (284,481) -- (284,505) .. controls (284,509.42) and (280.42,513) .. (276,513) -- (252,513) .. controls (247.58,513) and (244,509.42) .. (244,505) -- cycle ;
			\draw  [fill={rgb, 255:red, 255; green, 255; blue, 255 }  ,fill opacity=1 ] (336,481) .. controls (336,476.58) and (339.58,473) .. (344,473) -- (368,473) .. controls (372.42,473) and (376,476.58) .. (376,481) -- (376,505) .. controls (376,509.42) and (372.42,513) .. (368,513) -- (344,513) .. controls (339.58,513) and (336,509.42) .. (336,505) -- cycle ;
			\draw    (264,513) -- (361,590) ;
			\draw    (356,513) -- (266,587) ;
			\draw    (264,513) -- (266,587) ;
			\draw  [fill={rgb, 255:red, 0; green, 0; blue, 0 }  ,fill opacity=1 ] (439,588) .. controls (439,585.24) and (441.24,583) .. (444,583) .. controls (446.76,583) and (449,585.24) .. (449,588) .. controls (449,590.76) and (446.76,593) .. (444,593) .. controls (441.24,593) and (439,590.76) .. (439,588) -- cycle ;
			\draw  [fill={rgb, 255:red, 0; green, 0; blue, 0 }  ,fill opacity=1 ] (534,591) .. controls (534,588.24) and (536.24,586) .. (539,586) .. controls (541.76,586) and (544,588.24) .. (544,591) .. controls (544,593.76) and (541.76,596) .. (539,596) .. controls (536.24,596) and (534,593.76) .. (534,591) -- cycle ;
			\draw  [fill={rgb, 255:red, 255; green, 255; blue, 255 }  ,fill opacity=1 ] (422,481) .. controls (422,476.58) and (425.58,473) .. (430,473) -- (454,473) .. controls (458.42,473) and (462,476.58) .. (462,481) -- (462,505) .. controls (462,509.42) and (458.42,513) .. (454,513) -- (430,513) .. controls (425.58,513) and (422,509.42) .. (422,505) -- cycle ;
			\draw  [fill={rgb, 255:red, 255; green, 255; blue, 255 }  ,fill opacity=1 ] (514,480) .. controls (514,475.58) and (517.58,472) .. (522,472) -- (546,472) .. controls (550.42,472) and (554,475.58) .. (554,480) -- (554,504) .. controls (554,508.42) and (550.42,512) .. (546,512) -- (522,512) .. controls (517.58,512) and (514,508.42) .. (514,504) -- cycle ;
			\draw    (442,514) -- (539,591) ;
			\draw    (534,512) -- (444,588) ;
			\draw    (442,514) -- (444,588) ;
			
			\draw (147,587) node [anchor=north west][inner sep=0.75pt]    {$...$};
			\draw (39,488) node [anchor=north west][inner sep=0.75pt]   [align=left] {S};
			\draw (41,584) node [anchor=north west][inner sep=0.75pt]   [align=left] {T};
			\draw (102,487) node [anchor=north west][inner sep=0.75pt]    {$X_{1}$};
			\draw (103,583) node [anchor=north west][inner sep=0.75pt]    {$Y_{1}$};
			\draw (200,583) node [anchor=north west][inner sep=0.75pt]    {$Y_{d}$};
			\draw (199,487) node [anchor=north west][inner sep=0.75pt]    {$X_{d}$};
			\draw (266,597) node [anchor=north west][inner sep=0.75pt]    {$z_{1,1}$};
			\draw (356,597) node [anchor=north west][inner sep=0.75pt]    {$z_{1,d}$};
			\draw (252,490) node [anchor=north west][inner sep=0.75pt]    {$W_{1,1}$};
			\draw (343,488) node [anchor=north west][inner sep=0.75pt]    {$W_{1,d}$};
			\draw (303,584) node [anchor=north west][inner sep=0.75pt]    {$...$};
			\draw (299,496) node [anchor=north west][inner sep=0.75pt]    {$...$};
			\draw (147,495) node [anchor=north west][inner sep=0.75pt]    {$...$};
			\draw (389,533) node [anchor=north west][inner sep=0.75pt]    {$...$};
			\draw (443,597) node [anchor=north west][inner sep=0.75pt]    {$z_{n,1}$};
			\draw (527,597) node [anchor=north west][inner sep=0.75pt]    {$z_{n,d}$};
			\draw (430,491) node [anchor=north west][inner sep=0.75pt]    {$W_{n,1}$};
			\draw (521,489) node [anchor=north west][inner sep=0.75pt]    {$W_{n,d}$};
			\draw (481,585) node [anchor=north west][inner sep=0.75pt]    {$...$};
			\draw (477,497) node [anchor=north west][inner sep=0.75pt]    {$...$};

			\end{tikzpicture}

		\end{center}
		\caption{The constructed graph $G$ in the proof of Theorem~\ref{thm:w1_p}.}
	\end{figure}
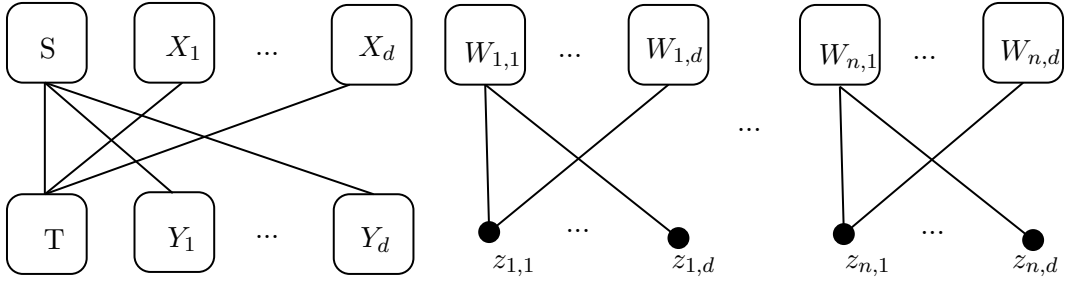
	\noindent\textbf{Proof of Sufficiency.} Suppose that $(\mathcal{I},\{\mathbf{w}_i\}_{i\in \mathcal{I}}, \mathbf{B},m,d)$ is a yes-instance. So, there exists a subset $I\subseteq \mathcal{I}$ where $|I|=m$ and for each $j\in [d]$, $\sum_{i\in I} w_{i,j}= B_j$. Now, we provide a $2$-coloring of $G$ which is $\ell$-fair. Define the color classes $C_1,C_2$ as follows.
	
	\begin{align*}
	C_1&=   (\bigcup_{i\in I, j\in [d]} W_{i,j}\cup X_j) \cup \{z_{i,j}: i\in \mathcal{I}\setminus I, j\in [d]\} \cup S,\\
	C_2&=   (\bigcup_{i\in \mathcal{I}\setminus I, j\in [d]} W_{i,j}\cup Y_j) \cup \{z_{i,j}: i\in I, j\in [d]\} \cup T.
	\end{align*}
	It is clear that $C_1$ and $C_2$ are two stable sets. Now, for each $j\in [d]$, we have 
	\begin{align*}
	|C_1\cap V_j|&= \sum_{i\in I} w_{i,j} + |X_j|+ |\mathcal{I}\setminus I|= \sum_{i\in I} w_{i,j} + B-B_j+ n-m= B+n-m=|S|,\\
	|C_2\cap V_j|&= \sum_{i\in \mathcal{I}\setminus I} w_{i,j} + |Y_j|+ |I|= \sum_{i\in \mathcal{I}\setminus I} w_{i,j} + B_j+A-A_j+ m= A+m=|T|.
	\end{align*}
	Also, $|C_1\cap V_{d+1}|=|S|$ and $|C_2\cap V_{d+1}|=|T|$. Therefore, $(C_1,C_2)$ is $0$-fair. \\
	
	\noindent\textbf{Proof of Necessity.} Now, suppose that $(C_1,C_2)$ are the color classes of a $2$-coloring for $G$ which is $0$-fair. Since $C_1$ and $C_2$ are stable sets, $S\cup \cup_{j\in [d]} X_j$ and $T\cup \cup_{j\in [d]} Y_j$ are in different color classes. Without loss of generality, suppose that $S\cup \cup_{j\in [d]} X_j \subseteq C_1$ and $T\cup \cup_{j\in [d]} Y_j\subseteq C_2$. Also, for each $i\in \mathcal{I}$, $\cup_{j\in [d]}W_{i,j}$ is a subset of one color class. Now, define 
	\[I=\{i\in \mathcal{I}: \cup_{j\in [d]}W_{i,j} \subseteq C_1\}. \]
	It is clear that for each $i\in \mathcal{I}\setminus I$, we have $\{z_{i,j}: j\in [d]\}\subseteq C_1$. 
	Since the coloring is $0$-fair, for each $j\in [d]$, we have
	\[
	|C_1\cap V_j|=  \sum_{i\in I} w_{i,j} + |X_j|+ |\mathcal{I}\setminus I|= \sum_{i\in I} w_{i,j} + B-B_j+ n-m= |C_1\cap V_{d+1}|= |S|= B+n-m.
	\]
	This implies that $ \sum_{i\in I} w_{i,j}= B_j$, as desired.
	
	Finally, note that $w_{i,j}$ and $B_j$'s are given in unary encoding and  $p=d+1$. So, this is a parameterized reduction. Moreover, if the graph $G$ is the disjoint union of $n+1$ complete bipartite graphs, then the modular width of $G$ is equal to two. Also, if each connected component of $G$ is a double-star, then $G$ is a forest whose each component is of depth $two$.
	This completes the proof. 
\end{proof}

\begin{theorem} \label{thm:xp}
	\fc\ is in XP with respect to $p$ whenever $k=2$ and can be solved in $O^*((2|V(G)|)^{p^2})$.
\end{theorem}
\begin{proof}
	Let $(G;V_1,\ldots, V_p;k,\ell)$ be an instance of \fc\ where $k=2$. The graph $G$ is bipartite, otherwise the answer is no. Suppose that $G$ has $n$ connected components $G_1,\ldots, G_n$ with bipartitions $(A_1,B_1),\cdots, (A_n,B_n)$. Also, for each $i\in [n]$ and $j\in [p]$, let $|A_i\cap V_j|=a_{i,j}$ and $|B_i\cap V_j|= b_{i,j}$. In any 2-coloring of $G$, vertices of $A_i$ and $B_i$ are in different color classes. Therefore, there is an $\ell$-fair $2$-coloring for $G$, if and only if there exists a subset $I\subseteq [n]$, such that for every $j_1,j_2\in [p]$, 
	
	\begin{align}
	\left( \sum_{i\in I} a_{i,j_1}+\sum_{i\in [n]\setminus I} b_{i,j_1}\right)- \left( \sum_{i\in I} a_{i,j_2}+\sum_{i\in [n]\setminus I} b_{i,j_2}  \right) \leq \ell, \label{eq:1} \\
	\left( \sum_{i\in [n]\setminus I} a_{i,j_1}+\sum_{i\in I} b_{i,j_1}\right)- \left( \sum_{i\in [n]\setminus I} a_{i,j_2}+\sum_{i\in I} b_{i,j_2}  \right) \leq \ell.\label{eq:2}
	\end{align}
	
	For each $i\in [n]$ and $j\in [p]$, let $c_{i,j}=a_{i,j}-b_{i,j}$, $a_j=\sum_{i=1}^n a_{i,j}$ and  $b_j=\sum_{i=1}^n b_{i,j}$. Thus, \eqref{eq:1} and \eqref{eq:2} are equivalent to 
	\begin{align}
	a_{j_1}-a_{j_2}-\ell \leq \sum_{i\in I} (c_{i,j_1}-c_{i,j_2}) \leq \ell-b_{j_1}+b_{j_2}. 
	\end{align}
	Let $W=(w_{j_1,j_2})$ be a $p\times p$ matrix such that $w_{j,j}=0$ and $w_{j_1,j_2}\in  [-|V(G)|,|V(G)]$. Also, let $t\in [n]$. Then define the function 
	\[
	f(t,W)=\begin{cases}
	1 & \exists I\subseteq [t] \text{ s.t. }\forall j_1,j_2\in [p],\,  \sum_{i\in I} (c_{i,j_1}-c_{i,j_2})=w_{j_1,j_2},\\
	0 & \text{otherwise}
	\end{cases}
	\]
	We find the value of $f$ by a dynamic programming as follows. First, note that $f(1,W)=1$ if and only if either $W=O$ (i.e. $I=\emptyset$), or there is some $i\in [n]$ such that for every $j_1,j_2\in [p]$, $c_{i,j_1}-c_{i,j_2}=w_{j_1,j_2}$. This can be checked in $O(n)$. 
	
	Finally, for every $t\geq 2$,  if $f(t,W)=1$, then there exists $I\subseteq [t]$ such that for each $j_1,j_2\in [p]$,  $\sum_{i\in I} (c_{i,j_1}-c_{i,j_2})=w_{j_1,j_2}$. If $t\not\in I$, then $f(t-1,W)=1$. If $t\in I$, then $\sum_{i\in I\setminus \{t\}} (c_{i,j_1}-c_{i,j_2})=w_{j_1,j_2}- c_{t,j_1}+c_{t,j_2}$. Define the matrix $W'$, where $w'_{j,j}=0$ and $w'_{j_1,j_2}=w_{j_1,j_2}- c_{t,j_1}+c_{t,j_2}$. Then, we have $f(t-1,W')=1$. Conversely, if either $f(t-1,W)=1$ or $f(t-1,W')=1$, then $f(t,W)=1$. Thus, we can compute the value of $f(t,W)$ in terms of $f(t-1,W)$. 
	
	The size of the dynamic programming table is $n\times (2|V(G)|)^{\binom{p}{2}}$. Also, $W'$ and $c_{i,j}$'s can be constructed in $O(p^2)$ and $O(|V(G)|)$. Hence, the whole problem can be solved in time $O^*((2|V(G)|)^{\binom{p}{2}})$.
	
\end{proof}

\section{FPT Results}
\begin{theorem} \label{thm:nd}
	\fc\ is fixed parameter tractable with respect to the neighborhood diversity whenever $k=2$ and can be solved in time $O^*(2^{\nd})$ where $\nd$ is the neighborhood diversity of the input graph. 
\end{theorem}

\begin{proof}
	Let $(G;V_1,\ldots, V_p;k,\ell)$ be an instance of \fc\ where $k=2$. 
	Also, let $(U_1,\ldots, U_{\nd})$ be a partition of $V(G)$, such that $G[U_i]$ is either a clique or a stable set and for each $i\neq j$, $U_i$ is either complete or incomplete to $U_j$. Since $G$ is bipartite, if $U_i$ is a clique, then $|U_i|=2$ and $U_i$ is anticomplete to other $U_j$'s, $j\neq i$. Without loss of generality, suppose that the sets $U_1,\ldots, U_t$ induce a clique $K_2$ and $U_{t+1},\ldots, U_{\nd}$ are stable sets. Also, suppose that for each $i\in [t+1,\nd-1]$, $U_i$ is complete to some $U_j$, $j\neq i$ and $V_{\nd}$ is anticomplete to all $U_j$'s. In fact, vertices in $U_{\nd}$ are singleton vertices in $G$.  
	
	If $(A,B)$ is a 2-coloring of $G$, then for each $i\in [t+1,\nd-1]$, we have either $U_i\subseteq A$ or $U_i\subseteq  B$. Also, for each $i\in [1,t]$, if $U_i=\{u,v\}$, then either $u\in A, v\in B$, or $u\in A,v\in B$. Now, fix a 2-coloring $(A,B)$ for the induced graph on the set $\cup_{i=1}^{\nd-1}U_i$. There are at most $2^{\nd}$ number of these colorings. Also, for each $j\in [p]$, let $|A\cap V_j|=a_j$ and $|B\cap V_j|=b_j$ and $|U_{\nd} \cap V_j|=n_j$. Now, each singleton vertex in $U_{\nd}$ can independently join to $A$ or $B$ to construct a 2-coloring for $G$. Now, suppose that $x_j$ vertices of $U_{\nd} \cap V_j$ are assigned to $A$ and $n_j-x_j$ other vertices are assigned to $B$. Therefore, we have $|A\cap V_j|= a_j+x_j$ and $|B\cap V_j|=b_j+n_j-x_j$. In order to have a fair-coloring, we need the following conditions hold.
	\begin{align*}
	&a_{j_1}+x_{j_1}- (a_{j_2}+x_{j_2}) \leq \ell, \ \ \forall j_1,j_2\in [p],\\
	&b_{j_1}+n_{j_1}-x_{j_1}- (b_{j_2}+n_{j_2}-x_{j_2}) \leq \ell, \ \ \forall j_1,j_2\in [p],\\
	&0\leq x_j\leq n_j, \ \ \forall j\in [p].
	\end{align*}
	If we set $r_{j_1,j_2}= \ell -a_{j_1}+a_{j_2}$ and set $s_{j_1,j_2}= b_{j_1}+n_{j_1}-b_{j_2}-n_{j_2}-\ell$, then we have the following relaxed linear programming. 
	\begin{align}
	\text{LP1:}&\nonumber \\
	&s_{j_1,j_2}\leq x_{j_1}- x_{j_2}\leq r_{j_1,j_2},\ \ \forall j_1,j_2\in [p], \label{eq:con1}\\  
	&x_j\in \mathbb{R},\ 0\leq x_j\leq n_j, \ \ \forall j\in [p]. \label{eq:con2}
	\end{align}
	Now, we are going to show that LP1  has  a feasible solution if and only if it has an integral solution. This LP is very similar to the LP corresponding to the vertex cover problem except that in the vertex cover problem, each condition contains the sum of two variables, while here Conditions \eqref{eq:con1} contain the  difference of two variables. Suppose that LP1 has a feasible solution and let $P$ be the polytope of the solutions of LP1. We need to show that the vertices of $P$ are integral. Let $X$ be a vertex of $P$. So, $X$ is the unique solution of some independent equalities corresponding to a subset of Conditions \eqref{eq:con1} and \eqref{eq:con2}, i.e. there is a subset $I\subseteq [p]\times [p]$ and $J\subseteq [p]$ such that $X$ is the unique solution of the system of equations:
	\begin{align*}
	&x_{j_1}-x_{j_2}= t_{j_1,j_2},\quad (j_1,j_2)\in I,	\\
	&x_{j}=m_j, \quad j\in J,
	\end{align*}  
	where $t_{j_1,j_2}\in\{s_{j_1,j_2},r_{j_1,j_2}\}$ and $m_j\in\{0,n_j\}$.
	Now, define a graph $H$ on the vertex set $[p]$ with the edge set $I$. We claim that since these equalities are independent, the graph $H$ is acyclic. To see this, suppose that $j_1,j_2,j_3,\ldots,j_t,j_1$ is a cycle in $H$, then the equation corresponding to $(j_t,j_1)$ can be obtained from the sum of equations corresponding to $(j_1,j_2), \ldots, (j_{t-1},j_t)$ which is in contradiction with the independency of the equations. This shows that $H$ is acyclic and each component of $H$ is a tree. Since the solution is unique, each component of $H$ contains some variable $x_j$, $j\in J$ and so all variables in this component are uniquely determined. Since $m_j$ and $t_{j_1,j_2}$ are integers, the solution $X$ is integral. This shows that $LP1$ has an integral solution. 
	
	Now, for each $2$-coloring of $\cup_{i=1}^{\nd-1} U_i$, we solve LP1 in polynomial time. Therefore, the whole problem can be solved in time $O^*(2^{\nd})$. 
\end{proof}

\begin{theorem} \label{thm:p+nd}
	\fc\ is fixed parameter tractable with respect to the neighborhood diversity and the number of groups $p$ and can be solved in time $O^*(2^{O(2^{\nd} \nd^2 p\log p)})$ where $\nd$ is the neighborhood diversity of the input graph. 	
\end{theorem}
\begin{proof}
	Let $(G;V_1,\ldots, V_p;k,\ell)$ be an instance of \fc. 
	Also, suppose that $(U_1,\ldots, U_{\nd})$ is a partition of $V(G)$, such that $U_i$ is a stable set in $G$ for $i\in[t]$ and is a clique in $G$ for $i\in [t+1,\nd]$, and for each $i\neq j$, $U_i$ is either complete or incomplete to $U_j$. Also, suppose that $|U_i\cap V_j|=n_{i,j} $ for all $i\in [\nd]$ and $j\in [p]$. 
	
	A subset $I\subseteq [\nd]$ is called valid if for every $i,i'\in I$, $U_i$ is anticomplete to $U_{i'}$ in $G$. Also, let $\mathcal{I}$ be the set of all valid subsets of $[\nd]$. We are going to find a fair coloring $(C_1,\ldots, C_{k})$. A nonempty subset $C_j$ is called of type $I$, if for every $i\in I$, $C_j\cap U_i\neq \emptyset$ and for every $i\notin I$, $C_j\cap U_i=\emptyset$. For each type $I\in \mathcal{I}$, we define a variable $y_{_I}$ as the number of color classes $C_j$ of type $I$. Also, for each $j\in [p]$ and $I\in \mathcal{I}$ and $i\in I$, we define a variable $x_{i,j,I}$ as the number of vertices in $U_i\cap V_j$ which are located in a color class of type $I$. Now, consider the following integer programming. 
	
	\begin{align}
	ILP2: \nonumber &\\
	& \sum_{I\in \mathcal{I}} y_{_I}=k, \label{eq:a} \\
	& \sum_{I\in \mathcal{I}: I\ni  i} x_{i,j,I} =n_{i,j}, \quad \forall\ i\in [\nd],\ \forall\ j\in [p], \label{eq:b}\\
	& \sum_{j=1}^p x_{i,j,I} \geq y_{_I}, \quad \forall\ I\in \mathcal{I},\ \forall\ i\in I\cap [t], \label{eq:c} \\
	& \sum_{j=1}^p x_{i,j,I} =  y_{_I}, \quad \forall\ I\in \mathcal{I},\  \forall\ i\in I\cap [t+1,\nd],  \label{eq:d} \\
	& \sum_{i: i\in I} x_{i,j,I}- \sum_{i: i\in I} x_{i,j',I}\leq y_{_I} \ell, \quad \forall\ j,j'\in [p],\ \forall\ I\in \mathcal{I}, \label{eq:e}\\
	& \text{if } y_{_I}=0, \text{then } \forall j\in [p], \forall i\in  I, x_{i,j,I}=0,\quad  \forall\ I \in \mathcal{I}, \label{eq:f}\\
	& y_{_I},x_{i,j,I}\in \mathbb{N}\cup \{0\},\quad  \forall\ I\in \mathcal{I}, i\in I, j\in [p]. \label{eq:h}
	\end{align}

	Condition~\eqref{eq:a} guarantees that the number of color classes is equal to $k$. Condition~\eqref{eq:b} guarantees that $(C_1,\ldots, C_k)$ is a $k$-partition of $V(G)$. Condition~\eqref{eq:c} and \eqref{eq:d} guarantee that each color class intersects each $U_i$, $i\in I$, in at least one vertex if $U_i$ is a stable set ($i\in [t]$) and in exactly one vertex if $U_i$ is a clique ($i\in [t+1,nd]$). Condition~\eqref{eq:e} guarantees fairness of each color class. Condition~\eqref{eq:f} is trivial because if $y_I=0$, then we have no color class of type $I$. 
	
	First, suppose that we have a fair coloring $(C_1,\ldots, C_k)$. Define $y_{_I}$ as the number of color classes of type $I$. Also, for each $I\in \mathcal{I}$, $i\in I$ and $j\in [p]$, we define $x_{i,j,I}$ as the number of vertices in $U_i\cap V_j$ which are in a color class of type $I$. Conditions~\eqref{eq:a}, \eqref{eq:b}, \eqref{eq:f} and \eqref{eq:h} trivially hold. Conditions~\eqref{eq:c} and \eqref{eq:d} hold because if $U_i$, $i\in I$, is a stable set (resp. clique), then $U_i$ intersects each color class of type $I$ in at least (resp. exactly) one vertex. Since the coloring is $\ell$-fair, summing up the fairness inequalities over all color classes of type $I$, yields \eqref{eq:e}.
	Therefore, ILP2 has an integral solution. 
	
	For the converse, suppose that ILP2 has an integral solution. Then, we construct a fair coloring $(C_1,\ldots, C_k)$. First fix a type $I\in \mathcal{I}$. Now, define the matrix $A$ such that $A_{ij}=x_{i,j,I}$ and apply Lemma~\ref{lem:balanced} for the matrix $A$ and $b=y_{_I}$. So, $A$ can be written as the sum $A^1+\cdots+A^{y_{_I}}$ such that 
	for each row $i\in [\nd]$, each column $j\in [p]$, and each $s,s'\in [y_{_I}]$, we have 
	\begin{equation} \label{eq:row}
	|(A^s)_{i}- (A^{s'})_{i}|\leq 1,
	\end{equation}
	and
	\begin{equation} \label{eq:col}
	|(A^s)^{j}- (A^{s'})^{j}|\leq 1.
	\end{equation}
	For each matrix $A^s$, we construct a color class $C_s$ which has exactly $A^s_{i,j} $ vertices in $U_i\cap V_j$. If for $i\in I$, $U_i$ is a clique, then by Conditions~\eqref{eq:d} and \eqref{eq:row}, $C_s$ intersects $U_i$ in exactly one vertex. Also, $I$ is a valid set. So, $C_s$ is a stable set and we have a proper coloring for $G$ with $k$ colors. Finally, because of Conditions~\eqref{eq:e} and \eqref{eq:col} the coloring is $\ell$-fair. So, we have an $\ell$-fair proper coloring with $k$ colors. Hence, solving the problem is equivalent to solving ILP2.
	
	Now, ILP2 has exactly $2^{\nd}\times \nd\times p+ 2^{\nd}$ variables. Thus,  by Theorem~\ref{thm:lenstra}, we can solve ILP2 in $O^*(2^{O(2^{\nd} \nd^2 p\log p)})$.
\end{proof}

\section{Concluding Remark}
In this paper, the structural complexity landscape of the fair coloring problem has been investigated. It is proved that the boundary of W[1]-hardness lies in the number of groups $p$ even for very restrictive graph classes such as forests and graphs with modular-width equal to two. Also, we proved that when we add neighborhood diversity as a parameter, the problem becomes tractable. The complexity for some parameters remains unknown, which are listed below.
\begin{itemize}
	\item In Theorem~\ref{thm:w1_p}, we have proved that the problem is W[1]-hard with respect to the number of groups $p$ (even for forests and graphs with $\mw=2$), while Theorem~\ref{thm:p+nd} shows that the problem is FPT with respect to $(\nd,p)$ and Theorem~\ref{thm:nd} shows that it is FPT with respect to $\nd$ when $k=2$. So, it is natural to ask if the problem is FPT with respect to $\nd$ in general? An easier question is that if the problem is FPT with respect to the vertex cover number ($\vc$)?
	
	\item In this paper, we proved that the problem admits an XP algorithm with respect to the number of groups $p$ when the number of colors is fixed to $k=2$ (Theorem~\ref{thm:xp}). This gives rise to the question if the problem is in XP with respect to $p$ for a general number of colors $k$?
	
	\item  Finally, establishing a polynomial kernel or proving stronger lower bounds for the problem with respect to the mentioned parameters remains a challenging open direction.
\end{itemize}


\end{document}